\documentclass[letterpaper, 10 pt, conference]{ieeeconf}  

\IEEEoverridecommandlockouts                              

\usepackage{amssymb,latexsym,amsmath}
\usepackage{graphics}
\usepackage{subcaption}
\usepackage{graphicx}
\usepackage{epsfig}
\usepackage{float}
\usepackage{amsmath}
\usepackage{arydshln}
\usepackage{color}
\usepackage{url}
\newtheorem{theorem}{Theorem}
\newtheorem{assumption}{Assumption}
\newtheorem{theorem1}{Theorem1}
\newtheorem{lemma}[theorem1]{Lemma}

\usepackage{mathtools}
\newenvironment{definition}[1][Definition]{\begin{trivlist}
\item[\hskip \labelsep {\bfseries #1}]}{\end{trivlist}}

\title{\LARGE \bf
Modelling, Controllability and Gait Design for a Spherical Flexible Swimmer
}

\author{Sudin~Kadam, Ravi N.~Banavar, and~Vivek~Natarajan
\thanks{All the authors are with Systems and Control Engineering Department, Indian Institute of Technology Bombay, Mumbai, India.}
\thanks{e-mail: sudin@sc.iitb.ac.in, banavar@iitb.ac.in, 
vivek.natarajan@iitb.ac.in.}}

\begin{document}

\maketitle
\thispagestyle{empty}
\pagestyle{empty}

\begin{abstract}

This paper discusses modelling, controllability and gait design for a spherical flexible swimmer. We first present a kinematic model of a low Reynolds number spherical flexible swimming mechanism with periodic surface deformations in the radial and azimuthal directions. The model is then converted to a finite dimensional driftless, affine-in-control principal kinematic form by representing the surface deformations as a linear combination of finitely many Legendre polynomials.  A controllability analysis is then done for this swimmer to conclude that the swimmer is locally controllable on $\mathbb{R}^3$ for certain combinations of the Legendre  polynomials. The rates of the coefficients of the polynomials are considered as the control inputs for surface deformation. Finally, the Abelian nature of the structure group of the swimmer's configuration space is exploited to synthesize a curvature based gait for the spherical flexile swimmer and a rigid-link swimmer.
\end{abstract}

\section{INTRODUCTION}

Locomotion relates to a variety of movements resulting in transportation from one place to another, and is crucial to existential requirements of microbial and animal life. A vast majority of living organisms are found to perform undulatory swimming motion at microscopic scales. Reynolds number, which is the ratio of inertial forces to viscous forces acting on the body in fluid, at these micro-scales conditions is extremely low - to the order of $10^{-4}$. i.e. viscous forces highly dominate the motion. To get a relative sense of the numbers, the Reynolds number for a man swimming in water is of the order
of $10^4$, whereas that for a man trying to swim in honey is of the order of $10^{-3}$ \cite{cohen2010swimming}, \cite{najafi2004simple}.

The analysis of these biological and bio-inspired engineering mechanisms at microscopic scales has attracted considerable attention in the recent literature. Swimming of unicellular organisms is one of the most fundamental processes in biology. Mechanism of motion of sperm cells, microbes \cite{gray1955propulsion}, \cite{gaffney2011mammalian} is important in not only understanding the locomotion at micro scales but is also crucial to conceive biomimetic robots for applications in medicine as drug delivery, \cite{alouges2015can}, \cite{cho2014mini}. However, a lot of this research has been around the swimmers with slender, rigid links, see \cite{avron2008geometric},  \cite{melli2006motion}, \cite{passov2012dynamics}, \cite{tam2007optimal}. However, many of the micro-organisms observed in the nature have non-slender shapes and are flexible, see biological mechanisms shown in figures \ref{fig:Amoeba} and \ref{fig:cocci}.

\begin{figure}[!htb]
\minipage{0.22\textwidth}
\begin{center}
  \includegraphics[width=0.95\linewidth]{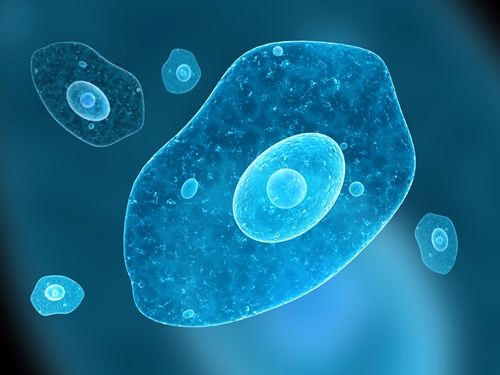}
  \caption{Amoeba \cite{amoeba}}\label{fig:Amoeba}
  \end{center}
\endminipage\hfill
\minipage{0.26\textwidth}
\begin{center}
  \includegraphics[width=0.62\linewidth]{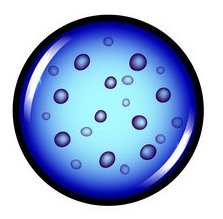}
  \caption{Micrococci bacteria \cite{microcci}}\label{fig:cocci}
\end{center}
\endminipage
\end{figure}
The work in this article pertains to a microswimmer which is spherical in shape and can have smooth, controlled surface deformations in the radial and tangential directions. From the mathematical modelling and control perspective, for the class of microswimming locomotion systems, the configuration space is amenable to the framework of a principal fiber bundle \cite{bloch1996nonholonomic}. In this approach, the configuration variables are naturally partitioned as the base and the group variables, and the former are, usually, fully actuated. The low Reynolds number effect gives rise to a principal kinematic form of the equations of  motion. Further, the shape space as well as the structure group of many of these systems is either the Special Euclidean group $SE(3)$ or  one  of  its  subgroups;  the  reader  is  referred  to  \cite{kelly1995geometric},  \cite{ostrowski1998geometric}  for many illustrative examples of such systems. A principal kinematic form of equation is obtained for this type of swimmers based on an approximate solution of the Stokes equations. The approach used in our work to model the flexible deformations of the swimmer is inspired by the theory and techniques presented in \cite{loheac2013controllability}, \cite{galdi2011introduction}.

In \cite{loheac2013controllability}, a kinematic model of the swimmer for radial deformations is obtained using a weak solution of the Stokes equation and a force operator followed by a controllability analysis using only two Legendre polynomials. \cite{shapere1987self} introduces a theory for infinite dimensional model of the radial and azimuthal deformations of this swimmer. The contribution of our work is that we present a model and controllability analysis for both radial as well as azimuthal axisymmetric deformations for any two successive Legendre polynomials. The controllability analysis in the present work thus gives a set of controls which are from a larger set of controls. Furthermore, using the Abelian natur of the Lie group of the system configuration space, we synthesize gait for the spherical flexible swimmer using the curvature technique.

The paper is organized as follows. In the next section we describe the swimmer's motion and revisit the Navier-Stokes equations along with the boundary conditions. In section III, we define the space of deformations and then obtain the kinematic model of the swimmer for radial and tangential of surface deformations. Section IV presents the controllability analysis of the swimmer based on the Chow's theorem. In section V, we present the curvature based gait synthesis technique for the spherical flexible swimmer using the Abelian nature of the group of the swimmer's configuration space. Here we also give an example of a rigid-linked microswimmer - a symmetric version of a popular rigid-linked low Reynolds number swimmer - the Purcell's swimmer.

\section{Background}
In this section, we summarize the modelling work presented in \cite{loheac2013controllability}.
\begin{subsection}{Problem Setup}
We consider a spherical swimmer surrounded by a viscous incompressible fluid filling the remaining part of the three dimensional space. We denote the swimmer's initial shape by $S_0$, the unit ball in $\mathbb{R}^3$. $S(t)$ denotes the domain in $\mathbb{R}^3$ occupied by the swimmer at time $t$. The fluid domain is thus given by $\Omega(t) = \mathbb{R}^3 / \overline{S(t)}$. Due to low Reynolds number conditions the velocity field $u(t, \cdot) : \Omega \to \mathbb{R}^3$ and the pressure field $p(t, \cdot): \Omega (t) \to \mathbb{R}$ satisfy the Stokes equations in $\Omega(t)$, \cite{galdi2011introduction}:
\begin{align}\label{eq:Stokes_equations}
-\mu \Delta u(t, \cdot ) + \nabla p(t, \cdot) &=0, \\
\nabla \cdot u(t, \cdot) &= 0
\end{align}
where $\Delta$ is the Laplacian operator, $\nabla$ is the gradient operator, $\mu \in \mathbb{R}^+$ is the fluid viscosity. The boundary conditions are
\begin{align}
\lim_{|x| \to \infty} u(t,x) = 0 	\qquad (t \geq 0), \\
u(t,\cdot) = v_s(t, \cdot) \qquad on \: \partial S(t),
\end{align}
where $v_s$ is the velocity of the swimmer on its boundary.
\begin{figure}[h!]
\centering
\includegraphics[scale=0.35]{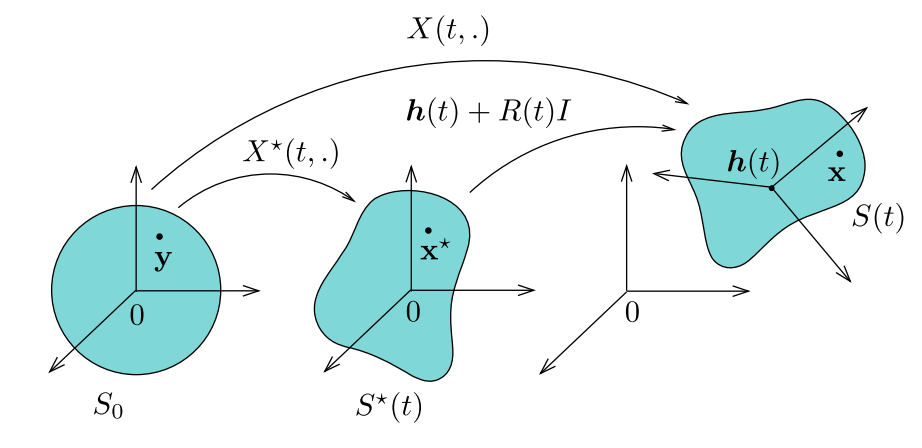}
\caption{Deformation and gross motion of the swimmer \cite{loheac2013controllability}}
\label{swimmer_shape_change}
\end{figure}

For every $v_0 \in H^{1/2} (\partial S) $ a unique weak solution $(v,\:p)$ of the Stokes problem in equations (1) and (2) exists for $ u \in H^{1/2}$ where $H^{1/2}$ is a fractional Sobloev space \cite{galdi2011introduction}. To obtain the kinematic model we define a bounded linear operator $\mathbb{F}(S) : H^{1/2} (\partial S) \to \mathbb{R}$. which 
\begin{equation}\label{force_operator}
\mathbb{F}(S) (v_0) = \int_{\Omega} \sigma (u,p) : \nabla u_z dx
\end{equation}
where $u_z$ denotes the component of fluid velocity $u$ along the axis of motion $v_0$ and $:$ denotes element wise multiplication, i, $\sigma (u,p) : \nabla u_z = \sum_{i,j} (\sigma (v,p))_{ij}) \times (\nabla u_z)_{ij}$. The force operator\footnote{We note that this is a bounded linear operator which is invariant under translation of $S$, i.e. for every $h \in \mathbb{R}^3$ and for every $v_0 \in H^{1/2} \left( \partial (S+h) \right)$, $\mathbb{F}(S+h) (v_0) = \mathbb{F}(S)(\tilde{v}_0)$, where $\tilde{v}_0 \in H^{1/2} (\partial S)$  and for every $x \in \partial S$, $\tilde{v}_0 (x) = v_0 (x+h)$.} $\mathbb{F}(S)$ associates to a given Dirichlet boundary conditions the total force exerted by the fluid.
\end{subsection}

\begin{subsection}{The configuration space and motion decomposition}
To analyze the motion of the swimmer, we decompose the entire motion into a rigid part and a flexible part (termed the shape change), see figure \ref{swimmer_shape_change}. The rigid part, characterized by the Special Euclidean group $SE(3) = SO(3) \ltimes \mathbb{R}^3$, consists of a translational displacement $h(t) \in \mathbb{R}^3$ and an orientation $R(t) \in SO(3)$ of the body frame attached at the geometric center of the body with respect to the inertial frame. The shape change, superimposed over 
the rigid part, is defined through the map $X^* : [0. \infty) \times S_0 \to \mathbb{R}^3$, which maps the original points in $S_0$ to deformed shape but with stationary center of mass. The total motion is characterized by $X(t, y): [0, \infty) \times S_0 \to \mathbb{R}^3$. For $y \in S_0,$ $X(t,y) = h(t) + R(t) X^*(t,y)$ is a pointwise map of every point $y$ on the sphere. Translational velocity of the center of mass of the swimmer is denoted by $\dot{h}(t)$, and the angular velocity of the body frame (at the center of mass) with respect to the inertial frame is denoted by $\omega$. The translational velocity $v$ of the swimmer for $x \in S(t)$ is then
%
%
\begin{equation}
\begin{aligned}
v(t,x) &= \underbrace{v_0 + \omega (t) \times \left( x-h(t) \right)}_{\text{Rigid component}} + \\ 
& \quad \: \underbrace{R(t) \frac{\partial X^*}{\partial t}(t,X^*(t,\cdot)^{-1}\left( R^T(t)\left(x-h(t)\right) \right)}_{\text{Shape component}}.
\end{aligned}
\end{equation}

 The velocity at the swimmer's outer surface defines the Dirichlet boundary conditions \cite{loheac2013controllability} for the exterior Stokes problem \eqref{eq:Stokes_equations}. We define the Cauchy stress tensor by $\sigma (u,p) = \mu \left( \nabla u + \nabla u^T \right) -p I_3$. The consequence of the low Reynolds number assumption is that the net forces and moments acting on the swimmer are always zero -
\begin{align*}
& \sum F = 0 \quad \implies \quad \int_{\partial S(t)} \sigma \textbf{n} d \Gamma =0 \\
& \sum M = 0 \quad \implies \quad \int_{\partial S(t)} (x - h) \times \sigma \textbf{n} d \Gamma = 0
\end{align*}
where \textbf{n} is the local outward normal to the swimmer surface.
\end{subsection}

\section{Axisymmetric surface deformations and a kinematic model}
The kinematic model is derived by solving the Stokes equation \eqref{eq:Stokes_equations}. In our work we employ the solution for the spherical swimmer for surface deformations in radial and azimuth direction from \cite{shapere1987self}. We consider spherical coordinates $(r, \theta, \phi) \in \mathbb{R}^+, \times [0, \pi], \times [0, 2\pi)$ and $e_z$ as the body frame axis of symmetry along which the swimmer performs the translational motion. 

\begin{assumption}
We limit the shape change of the swimmer to axisymmetric deformations in the radial $(e_R)$ and azimuthal $(e_{\theta})$ directions at any given point on the surface of the sphere, as shown in figure \ref{spherical_coordinates_Galdi}. \end{assumption}

\begin{figure}[h!]
\centering
\includegraphics[scale=0.35]{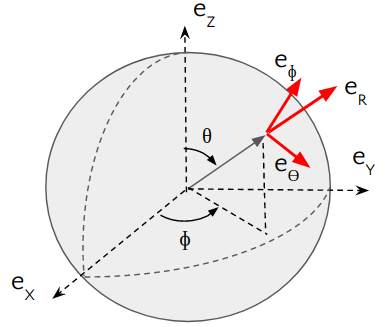}
\caption{The spherical swimmer and axes definition}
\label{spherical_coordinates_Galdi}
\end{figure}
%
%
We consider periodic axisymmetric shape changes composed of small deformations in the radial and azimuthal directions of a sphere of $S_0$ through the maps $r^*$ and $\theta^*$, respectively. For small number $\epsilon > 0$, the deformations can be written in terms of summation of $n^{th}$ order Legendre polynomials\footnote{The Legendre polynomials are solutions to the Legendre differential equation $(1-x^2) \frac{d^2 y}{dx^2} - 2x \frac{dy}{dx}+l(l+1)y=0, \: l \in \mathbb{N}$.}. At any given point on the swimmer surface, the deformed radius $r^*$ and the deformed angle from the axis of symmetry $\theta^*$ are defined as follows
\begin{align}
    r^*(\theta, t) = 1 + \epsilon \sum_{n=0}^{\infty} \alpha_n(t) P_n(\cos \theta), \label{radial_deformation}\\
    \theta^*(\theta, t) = \theta + \epsilon \sum_{n=0}^{\infty} \beta_n(t) V_n(\cos \theta) \label{azimuthal_deformation}
\end{align}
where $\alpha_n$ and $\beta_n$ are periodic functions of time whose rates $\dot{\alpha}_n, \dot{\beta}_n$ are the control variables, $P_n$ are the $n^{th}$ order Legendre polynomials as a function of cosine of $\theta$ and
\begin{equation}
    V_n(cos \theta) = \frac{1}{n+1} \frac{\partial}{\partial \theta} P_n(\cos \theta).
\end{equation}
%
In this case the radial and azimuthal velocities on the surface of the swimmer become
\begin{align}
    v_r(\theta, t) = \epsilon \sum_n \dot{\alpha}_n(t) P_n(\cos \theta), \label{radial_velocity} \\
    v_{\theta}(\theta, t) = \epsilon \sum_n \dot{\beta}_n(t) V_n(\cos \theta) \label{azimuthal_velocity}
\end{align}
%



\subsection{Kinematic model of the swimmer}
For the type of axisymmetric deformations defined in the previous section, we note that the motion for axisymmetric deformation, if at all, has to be in the $e_z \in \mathbb{R}^3$ direction only. We then employ the condition that at low Reynolds number conditions the inertia forces are negligible, and hence the net forces acting on the swimmer in any direction has to be zero. Hence, using the force operator, we can write the force balance equation in the $e_z$ direction as follows. 
\begin{equation}\label{eq:force_balance}
\begin{aligned}
 \dot{h} \mathbb{F}(X^*(t,\cdot)(S_0))e_z &+ \sum_{i=2}^{\infty} \left( \dot{\alpha}_i \mathbb{F}(X^*(t,\cdot)(S_0))P_i \right. \\ 
 & \left. + \dot{\beta}_i \mathbb{F}(X^*(t,\cdot)(S_0))V_i \right) = 0
\end{aligned}
\end{equation}
This leads to the expression for the net velocity of the swimmer in the $e_z$ direction as follows
\begin{equation}\label{eq:ode_swimmer}
\dot{h} = - \sum_{i=2}^{\infty} \frac{ \dot{\alpha}_i \mathbb{F}(X^*(t,\cdot)(S_0))P_i + \dot{\beta}_i \mathbb{F}(X^*(t,\cdot)(S_0))V_i }{\mathbb{F}(X^*(t,\cdot)(S_0))e_z}, \\
\end{equation}
According to \cite{shapere1987self}, based on the solution to the Stokes problem for the given radial and azimuthal deformations, the translational velocity $\dot{h}$ of the swimmer is obtained as the following series.
\begin{align}
    \dot{h} &= \epsilon^2 \sum_{n=2}^{\infty} \Bigg( \frac{(n+1)^2 \alpha_n\dot{\alpha}_{n+1} - (n^2-4n-2)\alpha_{n+1}\dot{\alpha}_n)}{(2n+1)(2n+3)} - \nonumber \\ 
    & \qquad \qquad \quad \frac{(n+1)(n+2)\alpha_n\dot{\beta}_{n+1} - n(n+1)\beta_{n+1}\dot{\alpha}_n}{(2n+1)(2n+3)} + \nonumber \\ 
    & \qquad \qquad \quad \frac{n(3n+2)\alpha_{n+1}\dot{\beta}_n + n(n+2)\beta_n \dot{\alpha}_{n+1}}{(2n+1)(2n+3)} - \nonumber \\
    & \qquad \qquad \quad \frac{n(n+2)\beta_n\dot{\beta}_{n+1} - n^2 \beta_{n+1}\dot{\beta}_n}{(2n+1)(2n+3)} \Bigg) + O(\epsilon^3) \label{solution_blake}    
\end{align}

We note that this equation is in a principal kinematic form where the shape space is infinite dimensional and the control inputs $\dot{\alpha}_i, \dot{\beta}_i, \: i = 2, 3, \dots, \infty$.

\section{Controllability of the flexible swimmer}
For systems of the form \eqref{solution_blake}, Chow's theorem is a useful tool to obtain a controllability result through which we show that the swimmer's center of mass can be translated from any given initial position to any given final position in $\mathbb{R}^3$. We first recall Chow's theorem. Let $m, n \in \mathbb{N}$ and let $(f_i)$ , $i= 1, \dots, m$ be $C^{\infty}$ vector fields on $\mathbb{R}^n$. Consider the control system,
\begin{equation}\label{eq:control_affine_system}
\dot{q} = \sum_{i=1}^m u_if_i(q)
\end{equation}
with input function $u = (u_i)_{i=1,m} \in C^{\infty} \left([0, \infty), \: B_{\mathbb{R}^m (0, r)} \right)$ where $B_{\mathbb{R}^m (0, r)}$ is an open ball for some $r > 0$. Let $\mathcal{O}$ be an open and connected set of $\mathbb{R}^m$. 
\begin{lemma}{(Chow's Theorem)}
The system \eqref{eq:control_affine_system} is locally controllable\footnote{The system \eqref{eq:control_affine_system} is said to be locally controllable at point $q$ in its configuration space, if the reachable set from $q$ contains an open neighborhood of $q$.} at $q \in \mathcal{O}$ if and only if

\begin{equation}
Lie_q \{ f_1, \dots, f_m \} = \mathbb{R}^n \quad \forall q \in \mathcal{O}
\end{equation}
\end{lemma}
We now examine the local controllability of the spherical flexible swimmer.
\begin{theorem}
The kinematic system given by \eqref{eq:ode_swimmer} is locally controllable at the spherical shape $S_0$ for $n \in \mathbb{N}, \: n>2$ when either of the following 2 conditions on the swimmer deformations are satisfied
\begin{enumerate}
    \item The radial deformations result from any two successive Legendre polynomials $P_n$ and $P_{n+1}$, or
    \item The azimuthal deformations result from any 2 successive Legendre polynomials $V_n$ and $V_{n+1}$
\end{enumerate}
\end{theorem}
\begin{proof}
Consider the case of purely radial deformations. In this case, near the shape $S_0$, \eqref{solution_blake} becomes
\begin{equation}
    \dot{h} = \epsilon^2 \sum_{n=2}^{\infty} \Bigg( \frac{(n+1)^2 \alpha_n\dot{\alpha}_{n+1} - (n^2-4n-2)\alpha_{n+1}\dot{\alpha}_n)}{(2n+1)(2n+3)} \Bigg)
\end{equation}
We observe that in case the deformations have to be from any 2 Legendre polynomials, for control vector fields to be non-trivial, they have to be from successive polynomials. Let us thus consider the control system in the principal kinematic form corresponding to any 2 successive radial deformations as follows-

\begin{align*}
\begin{bmatrix}
\dot{h} \\
\dot{\alpha}_n	\\
\dot{\alpha}_{n+1}
\end{bmatrix} &= f_n \dot{\alpha}_n + f_{n+1} \dot{\alpha}_{n+1} \\
 & = \begin{bmatrix}
\frac{- \epsilon^2(n^2-4n-2)\alpha_{n+1}}{(2n+1)(2n+3)} \\
1	\\
0	
\end{bmatrix}\dot{\alpha}_n + \begin{bmatrix}
\frac{\epsilon^2 (n+1)^2 \alpha_n}{(2n+1)(2n+3)} \\
0	\\
1	
\end{bmatrix}\dot{\alpha}_{n+1}
\end{align*}
We compute the Lie bracket of the control vector fields $f_n$ and $f_{n+1}$ to get
\begin{equation}
    [f_n, f_{n+1}] = \epsilon^4 \begin{bmatrix}
\frac{2n^2+6n+3}{4n^2+8n+3} \\
0	\\
0	
\end{bmatrix}
\end{equation}
We observe that for any $n \in \mathbb{N}$, $[f_n, f_{n+1}]$ is non-zero. Hence, the Lie algebra of the vector control vector fields $f_n$ and $f_{n+1}$ spans $\mathbb{R}^3$. Also, we note that $f_n, f_{n+1}$ and hence $[f_n, f_{n+1}]$ are $C^{\infty}$ functions. Hence, there exists $\epsilon_n > 0$ such that for every $(h, \: \alpha_n, \: \alpha_{n+1}) \in {0} \times (−\epsilon_n, \epsilon_n)^2$, the Lie algebra of $f_n$ and $f_{n+1}$ spans $\mathbb{R}^3$. This proves the first statement of the theorem. 

The second statement of the theorem is also proved in the similar way. The control system for azimuthal deformations from the successive values of $V_i$'s takes the following form -
\begin{align*}
\begin{bmatrix}
\dot{h} \\
\dot{\beta}_n	\\
\dot{\beta}_{n+1}
\end{bmatrix} &= g_n \dot{\beta}_n + g_{n+1} \dot{\beta}_{n+1} \\
& = \epsilon^2 \begin{bmatrix}
\frac{-n^2\beta_{n+1}}{(2n+1)(2n+3)} \\
1	\\
0	
\end{bmatrix}\dot{\beta}_n + \epsilon^2 \begin{bmatrix}
\frac{-n(n+2) \beta_n}{(2n+1)(2n+3)} \\
0	\\
1	
\end{bmatrix}\dot{\beta}_{n+1}
\end{align*}
We compute the Lie bracket of the control vector fields $g_n$ and $g_{n+1}$ to get
\begin{equation}
    [g_n, g_{n+1}] = \begin{bmatrix}
\frac{-2n^2-2n}{4n^2+8n+3} \\
0	\\
0	
\end{bmatrix}
\end{equation}
Like in the case of control vector fields for the radial deformations case, we observe that for any $n \in \mathbb{N}$, $[g_n, g_{n+1}]$ is non-zero. Hence, the Lie algebra of the vector control vector fields $g_n$ and $g_{n+1}$ spans $\mathbb{R}^3$. We again note that $g_n, g_{n+1}$ and hence $[g_n, g_{n+1}]$ are $C^{\infty}$ functions. Hence, there exists $\epsilon_n > 0$ such that for every $(h, \: \beta_n, \: \beta_{n+1}) \in {0} \times (−\epsilon_n, \epsilon_n)^2$, the Lie algebra of $f_n$ and $f_{n+1}$ spans $\mathbb{R}^3$. Thus the system is locally controllable at $S_0$ for any $2$ successive azimuthal deformations as well. This proves the second statement of the proof.

Additionally, if the shape control also allows radial and azimuthal deformations along the $e_x$ and $e_y$ directions, the swimmer is locally controllabile on $\mathbb{R}^3$.
\end{proof}

\section{Gait design}
A natural question that arises after the controllability test is whether one can synthesize controls which generate a desired trajectory of the system in its configuration space. In this section we look at an open loop gait design algorithm. We first present an algorithm to design a gait for the spherical flexible swimmer, and then extend this approach to the symmetric Purcell's swimmer - which is also a low Reynolds micro-swimming mechanism but with rigid links.

The algorithm rests on the notions of holonomy,
curvature\footnote{The curvature form $d\mathbb{A}$ is given by the exterior derivative of the local connection form $\mathbb{A}$.} and Stokes' theorem, and an essential feature is the Abelian property of the underlying group involved in describing the system. The holonomy in such systems relates to the gross displacements.
Stokes' theorem relates a contour integral to a volume integral, and the Abelian nature of the group helps synthesize the loop which encloses the desired volume under the surface formed by the curvature over the base space. 

\subsection{Gait design using Stokes' theorem}
A connection on a trivial principal bundle $Q = M \times G$ uniquely determines a system's trajectory in the full space from a loop $C$ in the base space. Recall that the system equation for a principally kinematic system is of the following form, see \cite{bloch1996nonholonomic}.
\begin{equation}\label{kinematics}
\xi = - \mathbb{A}(x) \dot{x}
\end{equation}
where $\xi$ is in the Lie algebra of the group $G$, $x \in M$ which is the base space and $\dot{x} \in T_xM$. With $e$ as the identity of the group $G$, integration of this equation with the initial condition $g(0) = e$ gives us the geometric phase corresponding to the closed curve $C \subset M$, where the geometric phase of a system evolving on a principal bundle is defined as follows.

\begin{definition}
A geometric phase or holonomy of a closed curve $c : [0,1]\:\rightarrow M$ is the net change in the group variable determined by the horizontal lift of $c$ \cite{kelly1995geometric}.

\begin{figure}[h!]
\centering
\includegraphics[scale=0.25]{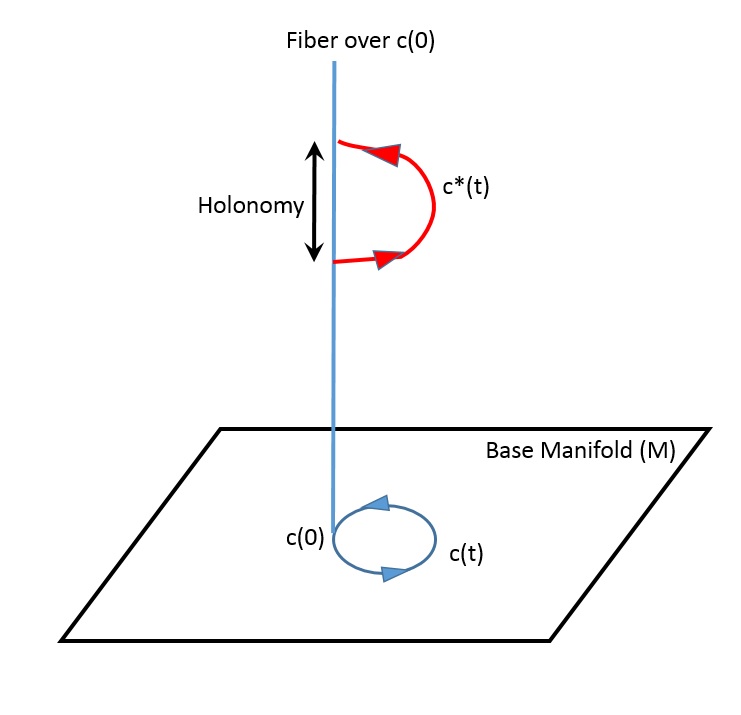} 
\caption{Horizontal lift of a base curve, Holonomy}
\label{Horizontal lift and holonomy}
\end{figure}
\end{definition}

If system is controllable, the following algorithm from \cite{kelly1995geometric} can be used to steer the system from any initial configuration $q(t_0) \in Q$ to any other configuration $q(t_f) \in Q$. 

\begin{itemize}\label{algo}
\item Choose a time $T_1 > 0$ and a path $x(.)$ such that $x(0) = x_0$ and  $x(T_1) = x_{T_1}$. Let $g(T_1)$ be the corresponding value of the group variable at the fixed time $T_1$, as determined by the system kinematics.

\item Choose a closed one dimensional curve $C \subset M$ such that the geometric phase associated with $C$ is given by $g_f(g(T_1))^{-1} \in G$
\end{itemize}

In general, the system equation $\eqref{kinematics}$ is difficult to integrate in a closed form. However, if the structure group $G$ is an Abelian Lie group, then the integral can be simplified considerably. Assume that $x(.) \in C$ is chosen such that $x$ traverses $C$ in unit time, Then, in the Abelian case, the solution to the equation $\eqref{kinematics}$ can be written as 
\begin{equation}\label{eqn_holonomy}
g(1) = exp\:(\int_0^1 (-\mathbb{A}(x).\dot{x})dt).g(0)
\end{equation}
where the exponential map $exp : \mathfrak{g} \rightarrow G$ is defined as a mapping from the Lie algebra $\mathfrak{g}$ to the Lie group $G$. That is, we get the net motion along the fiber by integrating the Lie algebra valued quantity $\mathbb{A}(x)\dot{x}$ along the base curve and then  taking its exponential.  We can further simplify this line integral by using Stokes' theorem that states:
{\it For a differential $k$-form $\omega$ and $S \subset M$ a $p$-dimensional region bounded by a curve $C = \partial S$}:
\begin{equation}\label{stokes_theorem}
\int_{\partial S} \omega = \int_S d \omega
\end{equation}
We can rewrite equation \eqref{eqn_holonomy} as
\begin{equation}
g(1) = exp \: (-\int_C \mathbb{A}(x))
\end{equation}
Then by applying Stokes' theorem, the net holonomy is
\begin{equation}
g(1) = exp(- \int_C \mathbb{A}(x)) g(0) = exp(\int \int_S d\mathbb{A}(s))g(0)
\end{equation}
That is, for systems on Abelian principal bundles the holonomy associated with the base space curve $C$ is same as the appropriately weighted area of any surface bounded by $C$.

\subsection{Gait design for the spherical swimmer}
We now apply the gait synthesis technique of the previous section to the flexible swimmer, since its Lie group is Abelian. Using the generalized Stokes' theorem, it can be shown that the volume enclosed by the a shape loop and the curvature surface of the local form of connection gives the displacement of the locomoting body under the cyclic shape change. We consider the spherical swimmer model \eqref{kinematics} evolves on a trivial principal bundle with Lie group $G =\mathbb{R}$ and the shape space $M$ as $\mathbb{R} \times \mathbb{R}$. The equation \eqref{eq:control_affine_system} is also in the principal kinematic form \eqref{kinematics} with the connection form given as
\begin{align}
\mathbb{A}(\alpha_2, \alpha_3) = & - \epsilon^2 \sum_{n=2,3} \frac{(n^2-4n-2)\alpha_{n+1}}{(2n+1)(2n+3)}d\alpha_n + \nonumber \\ 
& \frac{\epsilon^2 (n+1)^2 \alpha_n}{(2n+1)(2n+3)}d  \alpha_{n+1}
\end{align}

We consider radial periodic shape changes composed of small deformation using the second and third Legendre polynomials $P_2(x) = \frac{1}{2}(3x^2-1), P_3(x) =\frac{1}{2}(5x^3-3x)$. 
%
The same approach can be used to construct the gaits for azimuthal deformations. Using the expression for $M_(i,j)$ for $i,j \in {1,2}$, we now compute the curvature $2$-form for the spherical swimmer by taking the exterior derivative of the connection form as follows

\begin{align*}
d\mathbb{A} &= \frac{\partial \mathbb{A}_1}{\partial \alpha_1} d\alpha_1 \wedge d \alpha_2 + \frac{\partial \mathbb{A}_2}{\partial \alpha_1} d\alpha_2 \wedge d \alpha_1 \\
&= \left( \frac{\partial \mathbb{A}_1}{\partial \alpha_1} - \frac{\partial \mathbb{A}_2}{\partial \alpha_2} \right) d\alpha_1 \wedge d \alpha_2 \\
& = \left( \frac{9}{35} - \frac{6}{35} \right) d\alpha_1 \wedge d \alpha_2 \quad \text{using \eqref{solution_blake}}\\
& = \frac{3}{35} d\alpha_1 \wedge d\alpha_2
\end{align*}

We note that the curvature for the spherical flexible swimmer is constant.

\subsubsection{Algorithm}
We now present an algorithm for the open loop gait design for two different situations, to achieve commanded displacement in shape and group space. We first identify the maximum holonomy that the system can achieve by performing a loop in the shape space defined by the bounds on the shape variables. We denote this curve by $L_{max}$. In our case we define these bounds by restricting the absolute value of the coefficients by $0.2$, which limit the shape space within a square of length of sides as $0.4$. 
The maximum holonomy achieved by staying within these bounds is denoted by $H_{max}$. Since the curvature of the spherical swimmer is a constant, $H_{max}$ with these bounds on $\alpha_1,\: \alpha_2$ is $d\mathbb{A} \times 0.4 \times 0.4 = \frac{3}{35} \times 0.4 \times 0.4 = 0.0135m$
\begin{enumerate}
\item {\it Condition 1}: Initial and final shapes are identical
\begin{itemize}
\item If $H_{max}$ is positive, choose the direction of the loop traversal to be anticlockwise in the region of positive curvature
\item We then divide the task of achieving total holonomy as follows, 
\begin{equation}
H_{comm} = m \: H_{max} + H_{comm}\: modulo \: H_{max}
\end{equation}
where $m \in \mathbb{N}$.
\item Find that loop in the shape space which starts from the initial point of the system in the shape space and yields holonomy as $H_{comm}\: modulo \: H_{max}$. We denote this loop by $L_1$
\item \textbf{Maneuver 1:} If $m >0$, that is, if $H_{comm} \geq H_{max}$, from the given initial shape we do a shape maneuver to reach any convenient point on the maximal holonomy shape loop. We denote this curve by $C_1$
\item \textbf{Maneuver 2:} Then we perform the trace the $L_{max}$ loop $m$ times in appropriate direction
\item \textbf{Maneuver 3:} Then we trace again the curve $C_1$ in the direction opposite to that in maneuver 1 to bring the swimmer to the initial shape
\item \textbf{Maneuver 4:} Trace loop $L_{1}$ once in appropriate direction.
\end{itemize}
\item {\it Condition 2}: Initial and final shapes are different
\begin{itemize}
\item \textbf{Maneuver 1:} Do shape change to go along the straight line joining the initial and final shapes. The group displacement in this maneuver is denoted as $H_1$
\item \textbf{Maneuver 2:} Divide the total commanded holonomy as follows 
\begin{equation}
H_{comm} = m \: H_{max} + (H_{comm}-H_1)\: modulo \: H_{max}
\end{equation}
where $m \in \mathbb{N}$.
\item \textbf{Maneuver 3:} If $m >0$, that is, if $H_{comm} \geq H_{max}$, from the given initial shape we do a shape maneuver to reach any convenient point on the maximal holonomy shape loop $L_{max}$. We denote this curve by $C_3$
\item \textbf{Maneuver 4:} Then we trace the $L_{max}$ loop $m$ times in appropriate direction
\item \textbf{Maneuver 5:} Then we trace again the curve $C_1$ in the direction opposite to that in maneuver 1 to bring the swimmer to the initial shape
\end{itemize}
\end{enumerate}

\subsubsection{Results}
We illustrate an example for gait synthesis using the algorithm proposed. We restrict the shape variables' absolute value to $0.2$ and the swimmer starts with shape variables $\alpha_1= \alpha_2 =-0.2$, which corresponds to the initial shape as given in the figure \ref{case_2_initial_shape}. The swimmer is supposed to achieve the net displacement of $0.4m$ by performing shape maneuver by following the algorithm 1 above and coming back to the point in the shape space $\alpha_1 = \alpha_2 = -0.2$. Figures \ref{case_2_initial_shape} to \ref{case2_shape_space} shows the results.






\begin{figure}[h!]
\centering
\includegraphics[scale=0.3]{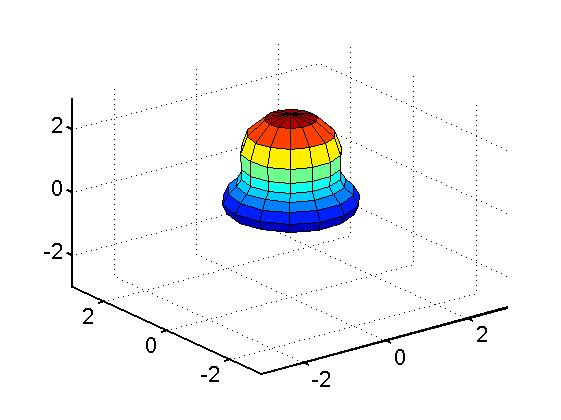}
\caption{Initial and final shape}
\label{case_2_initial_shape}
\end{figure}

\begin{figure}[h!]
\centering
\includegraphics[scale=0.2]{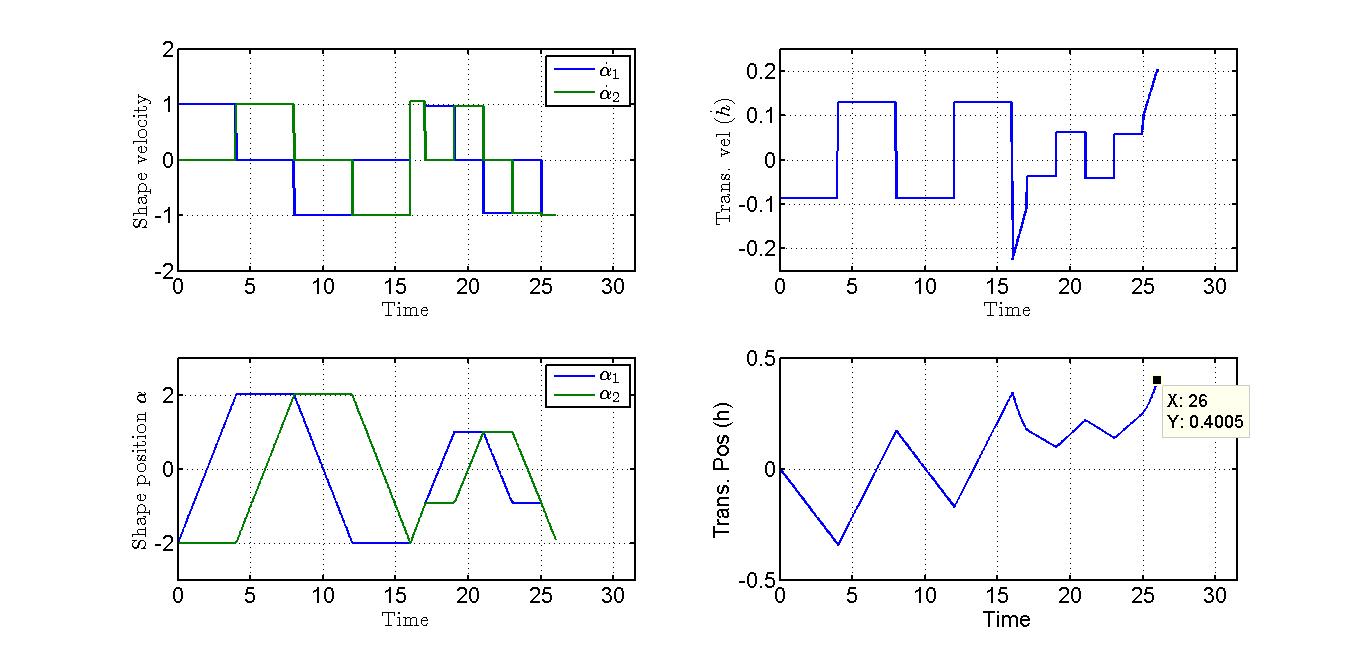}
\caption{Variation of the states and derivatives}
\label{case2_states_and_derivatives}
\end{figure}

\begin{figure}[h!]
\centering
\includegraphics[scale=0.25]{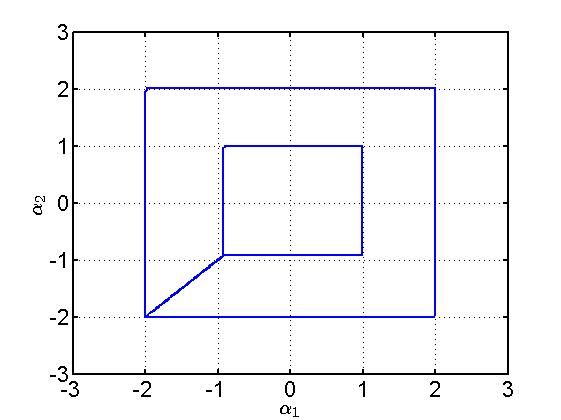}
\caption{The loop in the shape space}
\label{case2_shape_space}
\end{figure}

\subsection{Gait design for the symmetric Purcell's swimmer}
We  now mimic the procedure of the previous section to the symmetric Purcell's swimmer - a swimmer which has rigid links. The symmetric Purcell's swimmer is a variant of the 3-link planar Purcell's swimmer which has 4 limbs, see figure \ref{Symmetric_Purcells}. The limbs on one end of the central link are actuated symmetrically with respect to the axis along the length of the center link, i.e. $\alpha_1 = -\alpha_3$ and $\alpha_2 = \alpha_4$. 

\subsubsection{The model for the symmetric Purcell's swimmer}
The local form of connection for the symmetric Purcell's swimmer is a real-valued one-form on the base manifold $\mathbb{S}^1 \times \mathbb{S}^1$. Its expression for a viscous drag coefficient $k =1$ and for the limb length $L =1$ is obtained as follows, see \cite{kadam2015geometric} for the derivation,

\begin{equation}\label{symm_connection_form}
\mathbb{A}(x) = \frac{4\sin\alpha_1 d \alpha_1 -  4\sin\alpha_2 d \alpha_2 }{2 \sin^2 \alpha_1 + 2 \sin^2 \alpha_2 + 5}
\end{equation}
\begin{figure}{}
  \centering
  \includegraphics[scale=0.5]{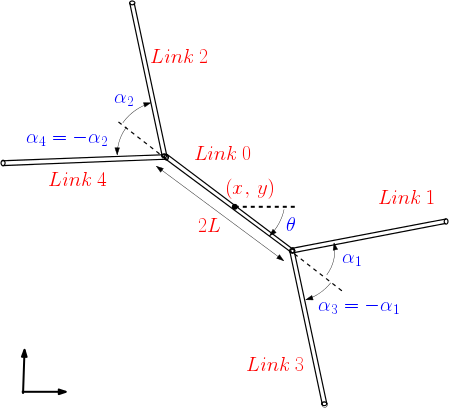}
  \caption{Symmetric Purcell's swimmer}
  \label{Symmetric_Purcells}
\end{figure}
The curvature $2$-form for the swimmer by taking the exterior derivative of the connection form as -
\begin{equation}
d\mathbb{A}=\displaystyle \displaystyle \frac{16 \sin \alpha_1 \sin \alpha_2 (\cos \alpha_1 - \cos \alpha_2)}{ (2 \sin^2 \alpha_1 + 2 \sin^2 \alpha_2 + 5)^2} \:\: d\alpha_1 \wedge d\alpha_2
\end{equation}
Figure \ref{sym_purc_curvature} shows the variation of the curvature as a surface, which along with a given base loop characterizes the net holonomy of the system.
\begin{figure}{}
  \centering
  \includegraphics[scale=0.22]{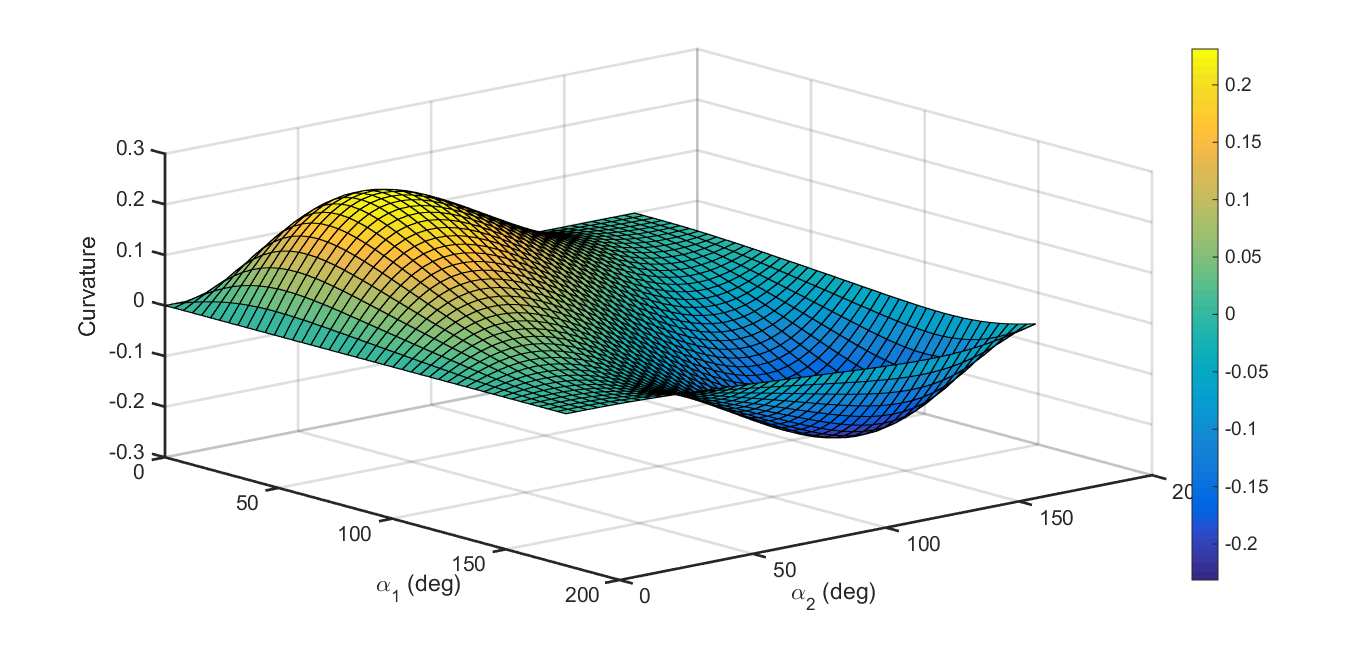}
  \caption{Curvature function}
  \label{sym_purc_curvature}
\end{figure}
We make a few observations from the curvature function for the swimmer as follows -
\begin{itemize}
\item The curvature attains zero value along the set $\mathcal{S}_1 = \{(\alpha_1, \alpha_2)\: | \: \alpha_1=0 \text{ or } \alpha_2 =0 \text{ or } \alpha_1 + \alpha_2 = \pi \}$
\item The curvature has a positive value in the region below the line $\alpha_1 = \alpha_2$, and vice versa.
\item The highest value of the curvature function is 0.2313, attained at $\alpha_1 = 44.12$ deg, $\alpha_2 = 44.12$ deg
\item The lowest value of the curvature function is -0.2313, attained at $\alpha_1 = 136.4$ deg, $\alpha_2 = 136.4$ deg
\item The volume enclosed by the triangular region in base space where the curvature is either positive or negative curvature region is $0.4039$
\end{itemize}

\subsubsection{Results}
We now synthesize an open loop gait of the symmetric Purcell's swimmer to achieve a commanded group displacement $H_{comm}$ of $+0.6$m. The swimmer with $k=1, L=1$ is considered. The objective here is to start and end at the point $(\alpha_1, \alpha_2) = (0,0) $ of the base space. We use algorithm  \ref{algo} here to achieve the objective as follows -

\begin{itemize}
\item We identify that the maximum dispacement $(H_{max})$ which can be achieved in a single base loop traversal is $0.41m$. The loop corresponding to the loop is the outer tranangle seen in \ref{Base_curve_for_the_multi_loop_holonomy}.
\item Identify $H_{comm} = 0.6m$ is greater than the maximum $(H_{max})=0.41m$.
\item If $H_{comm}$ is positive, choose the direction of the loop traversal to be anticlockwise in the region of positive curvature
\item We then divide the task of achieving total holonomy as follows, 
\begin{equation}
H_{comm} = m \: H_{max} + H_{comm}\: modulo \: H_{max}
\end{equation}
where $m$ is an integer. For the example in the figure \ref{Base_curve_for_the_multi_loop_holonomy} we get $m=1$
\item Identify that base circle $C_{circ}$ which yields holonomy as $H_{comm}\: modulo \: H_{max}$ on anticlockwise traversal
\item Starting from $(\alpha_1, \alpha_2) = (0,0)$ traverse a single loop along the triangle of the region of positive curvature to reach back to $(0,0)$.
\item From $(0,0)$ move along line $\alpha_1 = \alpha_2$ till a point of intersection with the circle $C_{circ}$
\item Fully traverse $C_{circ}$ in anticlockwise direction to reach back at the same point of intersection
\item Traverse back along the line $\alpha_1 = \alpha_2$ till the origin $(0,0) \in M$ is reached
\end{itemize}

\begin{figure}[!htb]
\centering
\includegraphics[scale=.22]{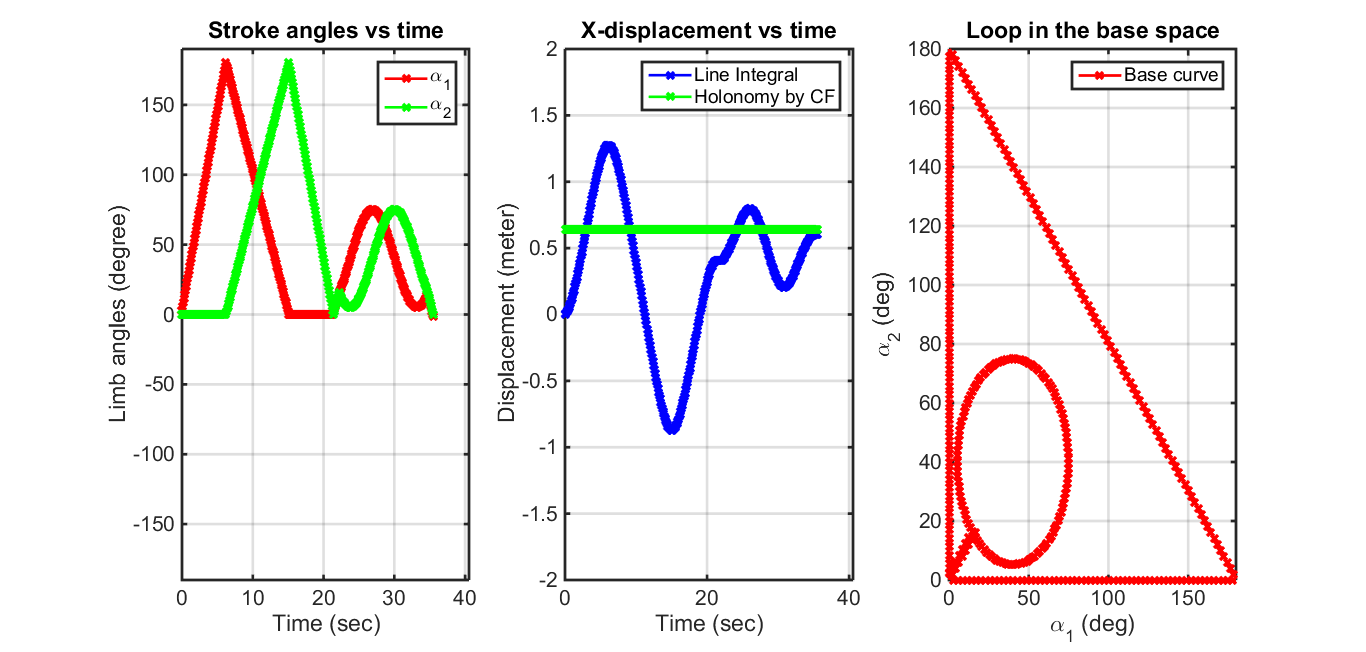}
\caption{Base curve for the multi loop holonomy}
\label{Base_curve_for_the_multi_loop_holonomy}
\end{figure}




\addtolength{\textheight}{-12cm}   





\bibliography{bib_spherical_swimmer}



\end{document}